\documentclass[12pt,reqno]{amsart}
\usepackage{amsfonts, amsmath, amssymb, amsthm, mathtools, mathrsfs, dsfont, tikz-cd, xcolor, braket}
\usepackage{tabularray}

\usepackage{geometry}
\usepackage[linktocpage, colorlinks=true, linkcolor={blue}, citecolor={blue}]{hyperref}

\usepackage[capitalise]{cleveref}
\crefname{theorem}{Theorem}{Theorems}
\crefname{lemma}{Lemma}{Lemmas}
\crefname{proposition}{Proposition}{Propositions}
\crefname{fact}{Fact}{Facts}
\crefname{corollary}{Corollary}{Corollaries}

\usepackage{enumitem}

\theoremstyle{definition}
\newtheorem{theorem}{Theorem}[section]
\newtheorem{corollary}[theorem]{Corollary}
\newtheorem{proposition}[theorem]{Proposition}
\newtheorem{definition}[theorem]{Definition}
\newtheorem{remark}[theorem]{Remark}
\newtheorem{example}[theorem]{Example}
\newtheorem{lemma}[theorem]{Lemma}
\newtheorem{prop}[theorem]{Proposition}
\newtheorem{cor}[theorem]{Corollary}
\newtheorem*{prob}{Problem}
\newtheorem{rem}[theorem]{Remark}
\newtheorem{fact}[theorem]{Fact}

\numberwithin{equation}{section}

\newcommand{\mb}{\mathbb}
\newcommand{\mc}{\mathcal}
\newcommand{\mf}{\mathfrak}
\newcommand{\lset}{\left\{ }
\newcommand{\rset}{\right\} }

\newcommand{\lpara}{\left(}
\newcommand{\rpara}{\right)}

\newcommand{\norm}[1]{\left\Vert #1 \right\Vert}
\newcommand{\su}{\mathfrak{su}}
\newcommand{\so}{\mathfrak{so}}

\newcommand{\kla}{\left ( }
\newcommand{\mer}{\right ) }
\newcommand{\ketbra}[1]{|{#1}\rangle\langle{#1}|}
\newcommand{\kl}{\pl \le \pl}
\newcommand{\gl}{\pl \ge \pl}
\newcommand{\lel}{\pl = \pl}
\newcommand{\ez}{{\mathbb E}}
\newcommand{\rz}{{\mathbb R}}
\newcommand{\zz}{{\mathbb Z}}
\newcommand{\cz}{{\mathbb C}}
\newcommand{\ten}{\otimes}
\newcommand{\pl}{\hspace{.1cm}}
\newcommand{\hz}{\vspace{0.5cm}}
\newcommand{\al}{\alpha}
\newcommand{\om}{\omega}
\newcommand{\si}{\sigma}
\newcommand{\la}{\lambda}
\newcommand{\eps}{\varepsilon}
\newcommand{\mz}{\mathbb{M}}
\newcommand{\lan}{\langle}
\newcommand{\ran}{\rangle}

\DeclareMathOperator{\tr}{tr}
\DeclareMathOperator{\id}{id}

\DeclareMathOperator{\ad}{ad}
\DeclareMathOperator{\Ad}{Ad}
\DeclareMathOperator{\diag}{diag}
\DeclareMathOperator{\Lin}{Lin}
\DeclareMathOperator{\Ham}{Ham}
\DeclareMathOperator{\lie}{lie}
\DeclareMathOperator{\cha}{Cha}

\title[One Hamiltonian Suffices]{
Avoiding exponentially large groups with open quantum system technology}
\author{Jihong Cai}
\email{jihongc2@illinois.edu}
\author{Advith Govindarajan}
\email{agovind2@illinois.edu}
\author{Marius Junge}
\email{mjunge@illinois.edu}
\thanks{MJ is partially supported by NSF DMS 2247114 and NSF DMS 2556195}

\address{Department of Mathematics, University of Illinois Urbana-Champaign}

\begin{document}
\maketitle
\begin{abstract}
    We propose a novel approach to state preparation by embracing the full power of open quantum systems. Instead of working with the whole unitary group in $n$ qubits, we identify a small Lie group $G$ in $(n+3)$ qubits. The dimension of its Lie algebra is $\text{poly}(n)$. Every state can be approximated starting from any fixed input state by repeatedly using partial trace and state preparation in the environment alongside group operations. Moreover, we also identify a single (open system magic) interaction Hamiltonian whose unitary group can be combined with open system quantum technology mentioned above to achieve transitivity on the space of all densities. Similarly, we show that all channels on $n$ qubits can be approximated by using a larger environment, achieving channel universality. The interesting small groups we identify are derived from Lie algebras of local Pauli strings and lead to a new landscape of cheap and expensive states, densities, and channels.
\end{abstract}

\tableofcontents
\section{Introduction}
The theory of quantum information science studies open and closed quantum systems. From a physical point of view open quantum systems governed by an interaction Hamiltonian appear to be natural starting point. In the early stages of quantum information science considerable effort was required to extract  single qubit control from  suitable Hamiltonian \cite{lloydviola, verstraete2008}. On the other hand, the progress in the theory of closed quantum systems in the last decades is undeniable: Circuits can be executed on existing quantum devices \cite{qis},  and perform fault tolerant computation in experiments\cite{Luk}. From a computer science perspective new paradigms for computation add interesting layers in  complexity theory, see e.g. \cite{Cleve}.

The discovery of universal gate sets is one of the milestones of quantum information theory \cite{deutsch1989, Kitaev}. A particularly striking example is the Clifford group: Clifford circuits have a rich algebraic structure and admit efficient classical simulation \cite{gottesman}, yet they are not universal for quantum computation. Adjoining a single non-Clifford resource, such as the $T$-gate $T=\diag(1, e^{i\pi/4})$ changes the picture completely: Clifford+$T$ is universal \cite{boykin, bravyikitaev}.

A useful continuous analogue of a magic gate is therefore the following. Given a resource set $S$, call a Hamiltonian $H$ \textit{magic} when $\lie_\mb R(iS\cup\lset iH\rset)=\su(A)$. This viewpoint was developed systematically in \cite{jpv}. Several simple spin-chain examples already illustrate how little additional control may be needed:
\begin{enumerate}\label{mh}
    \item[(i)] $H=\sum_{j=0}^{n-1}(X_jX_{j+1}+Y_jY_{j+1}+Z_jZ_{j+1})$ is magic with resp. to $S=S_E=\{X_0,Z_0\}$;
    \item[(ii)] $H= Z_1Z_2+\sum_{j=0}^{n-1} X_jX_{j+1}$ is magic with resp. to  to $S=\{X_0,Z_0\}\cup \{Z_1,...,Z_n\}$;
    \item[(iii)] $H=X_0+Z_0+Z_1Z_2+\sum_{j=0}^{n-1} X_jX_{j+1}$ is magic with resp. to $S=S_A=\{Z_1,...,Z_n\}$. 
\end{enumerate}
The first example is a familiar instance of local controllability for a Heisenberg chain \cite{Burgarth2009}; related Pauli-string constructions, including (ii) and (iii), appear in \cite{jpv}. In each case, the resource set itself is small, but after adjoining the magic Hamiltonian the Lie algebra becomes the full, exponentially large algebra.

\hz 
The key idea of this paper is to replace the $T$-gates and magic Hamiltonian by open quantum system technology. 
\hz 

The use of open quantum systems as a resource has of course appeared in several forms. Lloyd and Viola \cite{lloydviola} and Verstraete, Wolf, and Cirac \cite{verstraete2008} showed that dissipation and irreversible operations can be harnessed for control and computation. More recently, one-qubit environment has appeared naturally in dissipative state-preparation algorithm \cite{dingchenlin}. Campbell's work shows another way in which passing from unitaries to channels changes the synthesis problem: mixture of unitary approximations can improve channel-level accuracy \cite{campbell}, while qDRIFT uses randomized elementary evolutions for Hamiltonian simulation \cite{campbell2}. Ancillary systems also play a striking role in symmetry-constrained controllability: Marvian and collaborators showed that a small number of ancilla can overcome strong restrictions imposed by locality and symmetry \cite{marvian, marvianliu}.
We continue to harness the power of open quantum system technology by 
addressing the following:
\begin{prob}
    {\it Is it possible to find a small Lie algebra on the combination of system and environment, while obtaining universal state preparation with the help small trace out and state preparation?}
\end{prob}

To make this problem more precise, let $E$ be an environment and let $G\subseteq U(EA)$ be the available unitary group. Preparing $ v \in E$, applying $U\in G$, and discarding $E$ produces the channel 
$$\Phi_U^v(\rho)=\tr_E(U(\ket{v}\bra{v}\otimes \rho)U^*).$$
For a set $V$ of preparable environment states, let $\mc C_V(G)$ be the closed class generated by the map $\Phi_U^v$ under composition. We say that  $(V, G)$ performs  \textit{density magic} when $\mc C_V(G)$ acts transitively on the density matrices $\mc D(A)$ of $A$. In particular, every mixed state can be prepared by a channel in $C_V(G)$ applied to $|0\ran\lan 0|^A$. 
In this model, the environment is not merely an ancillary register which participates coherently in a larger circuit. It may be discarded and prepared again between successive applications of element $G$, this is called reusable qubit  in quantum computation \cite{dingzhan}.  This elementary irreversible operation turns out to be a game changer.
This basic mechanism can already be seen at first order. For $v\in E$ one can define the functional $\om_{v}(x)=(v,xv)$ and then apply the slicing procedure 
\[ \om_v(H):= (\om_{v}\ten id)(H)
\in\mc L_{sa}(A)\] 
for a Hamiltonian in the Lie algebra generated by $iS$. 
The map $H\mapsto \pi_v(H)$ is linear, but not a Lie algebra homomorphism. In particular, 
\[[H, K]=0 \quad\not\Longrightarrow \quad[\om_v(H), \om_v(K)]=0.\]
Using the power of noncommutativity, we are able to make excessive use of this observation. 
We construct a set of generators $S\subset L(EA)$ with small Lie algebra $\mathfrak{g}$ and small Lie group $G(S)$ so that $\mc U_{C_V(G)}$, the group of unitaries $U$ such that $\Ad_{U}:\rho \mapsto U\rho U^*$ lies in $C_V(G)$, is $U(A)$.  

There is also a dissipative part to the story. We prove a general transitivity criterion which, in a convenient form, say that 
$$U(A) = U_{C_V(G)} ,\quad \text{and} \quad e^{tL_{a\otimes 1\ten\dots\ten1}}\in C_V(G) \quad \text{with} \quad  [a,a^*]\neq 0,$$
already implies that $\mc C_V(G)$ is state-transitive. Here, $L_a(\rho)=2a\rho a^*-a^*a\rho-\rho a^*a$ denotes the dissipative Lindblad generator of a single jump operator.
Thus arbitrary density preparation follows from effective coherent control together with a single non-normal local dissipative direction. The role of open system is twofold: first-order reduction creates new Hamiltonian directions, while second order reduction creates Lindbladian directions. The criterion builds on the controllability picture for Lindblad dynamics developed in \cite{cgj, bwsh16}.

Our results indicate that there is a tradeoff between the dimension of the small Lie group $G\subset SU(EA)$ versus the number of additional qubits and preparation states in the environment:
\begin{center}
\begin{figure}[h]
\begin{tabular}{c|c|c|c}
        & preparable state $V$ & coherent group $G\subseteq U(EA)$ & $\dim\lie(G)$\\
    \hline
    (1) &  $V=\{|0\ran\}$ & $G(S)=SU(EA)$ \pl, $S=S_E\cup S_A\cup\{H\}$ & exponential \\
    (2) &  $V=\{|000\ran,|\xi\xi\xi\ran\}$ & $G\subset \mathbb{T} \times SPIN(n+3)\times SPIN(n+3)$ & $O(n^2)$\\
    (3) &  $V=\{|0\ran,|1\ran\}$ & $G=\{e^{ith}| t\in \rz\}, h\in L(EA)$ & $1$ \\
    
\end{tabular}
\caption{Transitivity on the space of density with more vectors and smaller groups}
\label{main}
\end{figure}
\end{center}
Model (1) is the familiar closed-system picture. The second model (2) uses a refinement of the Jordan Wigner transform.  The vector $\xi=\frac{1}{\sqrt{6+2\sqrt{3}}}(\sqrt{3}+1, 1+i)^T$ is the $+1$-eigenvector of the symmetry  $\frac{X+Y+Z}{\sqrt{3}}$. Example (2) also allows to identify certain `cheap' operations given by the pulses of certain local Pauli strings 
\begin{align*}
    S&=  \{X_{-2}X_{-1}X_0X_1,Y_1Y_2,X_{-2}X_{-1}X_2X_3,...,X_{n-1}X_n,Y_{n-1}Y_n\} \\
    & \pl \pl  \pl \pl \pl \pl  \cup 
    \{Y_{-2}Y_{-1}Y_0Y_1,X_1X_2,Y_{-2}Y_{-1}Y_2Y_3,...,Y_{n-2}Y_{n-1},X_{n-1}X_n\} \cup \{Z_{-1}Z_1Z_2\} \pl. 
\end{align*} 
which generate the small group $G$. This example is motivated by Araki and Matsui's work on many body systems (see appendix in \cite{jpv}).
Model (3) is more extreme. The coherent control is reduced all the way to a single one-parameter group. The same interaction Hamiltonian $h$ is used in every round, while the environment is prepared in one of two fixed states. Yet the resulting reduced dynamics area again state-transitive. In this sense, \textbf{one Hamiltonian suffices!}

Although we have reduced the dimension of the required Lie algebra, the task of state preparation remains provably difficult. For this we define $l(\eps,G,V)=\inf \lset m\rset$ such that for $\eta \in A$ we can find $v(1),....,v(m)\in V$ and $U_1,...,U_m\in G$ such that 
 \[ \|\Phi_{U_1}^{v(1)}\cdots \Phi_{U_m}^{v(m)}(|0\dots0\ran \lan 0\dots0|)-|\eta\ran\lan \eta|\|_1< \eps \pl .\] 
\begin{theorem}
    Let $A=\cz^{2^n}$. Then  $l(\eps,G,V)\gl c(\eps)\frac{2^n}{n}$ for state preparation with model (2). For density preparation a lower bound for the length is $c(\eps) \frac{2^{2n}}{n^3(1+\ln n)}$ for model 2) and $c(\eps)\frac{2^n}{n^3(1+\ln n)}$ for the single magic Hamiltonian $h$.
\end{theorem}

General channel magic can easily be obtained as an application using more sophisticated trace outs. Recall that 
 \[ \Phi(\rho) \lel tr_{F}(U(|0\ran\lan 0|\ten \rho)U^*) \] 
can be obtained with an environment such that $|F|=|A|^2$.  
\begin{theorem}\label{chm}  Let $G$ be the group for $\hat{A}=FA$ in \cref{main}. Let $\hat{E}=EF$ the environment obtained by adding one qubit to $F$. Let $\hat{C}_V(G)$ be the closed class of channels obtained by applying unitaries in $G$, partial trace outs in $E$ and a final trace out in $F$, and closed under composition. Then $\hat{C}_V(G)$ is the class of all channels. 
\end{theorem}

The paper is organized as follows: In Section 2, we introduce the open-system framework and the channel classes generated by unitary evolution, state preparation, and partial trace. Section 3 gives general criteria for state transitivity and relates them to interaction Hamiltonians. Section 4 and 5 contain our two main constructions: a polynomial-dimensional Lie group generated by local Pauli strings, and a single open-system magic Hamiltonian. Section 6 studies approximation length and proves exponential lower bounds for the number of open system rounds required to prepare worst-case pure states and density matrices. Section 7 extends the framework to channel universality and channel complexity. We conclude in Section 8 with several open directions and the connection to variational quantum eigensolver (VQE) and the work of \cite{diaz2023, Larocca2025}.

\textbf{Acknowledgment}: The third named author wants to thank Luke Visser and Jason Pollack for inspiring conversations.
The first and the third named authors also thank 
Roy Araiza and Milad Marvian for many stimulating discussions.
Part of this research was performed while the third named author was visiting the Institute for Pure and Applied Mathematics (IPAM), which is supported by the National Science Foundation (Grant Nos. DMS-1925919 and DMS-2422832).

\section{Channel Class}
Our Hilbert space will consist of the system space $A \cong \mathbb (\mathbb C^2)^{\otimes n}$ and environment $E \cong (\mathbb C^2)^{\otimes m}$. Throughout most of the paper, we work with a small environment, typically with $m\in\lset 1,2,3\rset$. We use $\mc L(H)$ to represent the space of operators on the Hilbert space $H$ and $\mc L_{sa}(H)$ for self-adjoint operators. In order to capture open system properties, we have to work with channels and densities rather than unitaries and vectors:
\begin{definition} \mbox{}
\begin{enumerate}
    \item  A \textit{state} or \textit{density matrix} is a matrix $\rho \in \mc L(A)$ with $\tr(\rho) = 1$,  $\rho \geq 0$ and $\rho = \rho^*.$ We write $\mc D(A)$ for the space of all states.
    \item A \textit{channel}
    is a completely positive trace preserving linear map $\Phi: \mc L(A) \to \mc L(A).$
    \end{enumerate}
\end{definition}

We will place the environment registers before the system registers, so our total Hilbert space will be $EA$. Discussion of channel magic will be left to \cref{channelmagic}. Our problems will generally consist of finding transitive protocols on the space of density matrices or in approximating all channels $\Phi: \mc D(A) \to \mc D(A)$.

\begin{definition}
\mbox{}
    \begin{enumerate}
        \item Let $\mc C$ be a class of channels. $\mc C$ is called \textit{state transitive} if for any states $\rho, \sigma \in \mc D(A)$ we have $\exists \Phi \in \mc C$ with $\Phi(\rho) = \sigma.$
        
        \item Let $\mc C$ be a class of channels. We call $\mc C$ \textit{channel universal} if its topological closure $\overline{\mc C}$, where the closure is taken with respect to the diamond norm on superoperators, contains all channels.
    \end{enumerate}
\end{definition}
We refer the readers to \cite{Watrous_2018} for background on diamond norm. All  of the protocols in this paper will require a set of preparable states $V \subseteq E$, consisting of unit vectors, and a set of resources $S$. We require our resources to be given by self-adjoint matrices, $S \subseteq \mc L_{sa}(EA)$. In general the set $S$ will be given by Pauli strings. 
\begin{definition}[The Gate Sets] \mbox{}
\begin{enumerate} 
    \item To a resource set $S \subseteq \mc L_{sa}(EA)$ we associate the Lie algebra $\mf g(S) \subseteq  \mf u(EA)=i\mc L_{sa}(EA)$ by 
    \[
    \mf g(S) = \text{span}_{\mb{R}}\{ iS \cup [S, S] \cup i[[S, S],S] \cup i[S, [S,S]] \cup [[[S,S],S],S] \dots \}.
    \] This is the smallest Lie algebra in $i\mc L_{sa}(EA)$ containing $iS$. 

    \item To a generator set $S$ we associate the group $G(S) \subseteq SU(EA)$ as the Lie group generated by $\mf g(S)$. 
\end{enumerate}
\end{definition}
In general $\mf g(S)$ will be inside the lie algebra $\mf{su}(EA) = \{s \in \mc L(EA): s = -s^*, \, \tr(s) = 0\}.$ We think of the group $G(S)$ as representing which operations are allowed on the the global Hilbert space $EA$. Every unitary in $G(S)$ can be approximated by unitaries in the set $\exp(itS)$ thanks to the celebrated Chow-Rachevskii Theorem, see \cite{LeDonne2021SubRiemannian}. Iterated commutators are obtained from the zigzagging formula 
 $$e^{iat}e^{ibt}e^{-iat}e^{-ibt} = e^{-[a,b]t^2+ O(t^3)}.$$

Starting from the set of unitary gates $G(S) \subseteq U(EA)$ we can generate a class of channels.
To do so we require a partial trace $tr_E$ and state preparation $\rho \mapsto \ket v  \bra v \otimes  \rho: \mc D(A) \to \mc D(EA)$ for some $\ket v \in E$.  

\begin{definition}[The Channel Set] \mbox{}
\begin{enumerate}
    \item Let $V \subseteq E$ be a set of preparable states and let $G \subseteq U(EA)$ be some unitary group. To any $\ket v \in V$ and $U \in G$ we define the channel $$\Phi^v_U(\rho) = \tr_E(U \lpara \ketbra{v} \otimes \rho \rpara U^*): \mc D(A) \to \mc D(A).$$

    \item The \textit{channel class} $\mc C_V(G)$ associated to the pair $(V,G)$ is the smallest set of channels $\mc D(A) \to \mc D(A)$ which contains the set $\{\Phi_{v,U}: \ket v \in V, U \in G\}$, is topologically closed and closed under composition.
     
\end{enumerate}
\end{definition}

To understand the class $\mc C_V(G)$, we will analyze its Lindbladian generators.  The structure of the generators of channels was studied by Lindblad \cite{Lindblad1976} and independently by Gorini, Kossakowski and Sudarshan \cite{GKS}. We recall some of the results we will be using in this paper:
\begin{definition}
    A \textit{Lindbladian} is a map $L: \mc L(A) \to \mc L(A)$ which generates a semigroup of quantum channels, that is to say $\exp(tL)$ is a quantum channel for all $t \geq 0$.
\end{definition}
\begin{fact}
    A Lindbladian is always of the following form:
    $$L(\rho) = i[H,\rho] + \sum_j^r \gamma_j  \lpara 2 a_j \rho a_j^* - a_j^* a_j \rho - \rho a_j^* a_j  \rpara$$
    where $H=H^*$, $\gamma_j \geq 0$, and $a_j \in \mc L(A)$. To isolate the dissipative part of the Lindbladian we will write:
    $$L_{a_j} := 2 a_j \rho a_j^* - a_j^* a_j \rho - \rho a_j^* a_j  $$
    so that we have
    $$L = i \ad_H + \sum_i^r \gamma_j L_{a_j}.$$ 
The operators $a_i$ are called \textit{jump operators}. Here and in the following  we use the notations $$\ad_{H}(\rho)=[H,\rho]\qquad \Ad_U(\rho)=U\rho U^*$$ for the Hamiltonian part and unitary conjugation channel.
\end{fact} 
The channel class will contain some unitary conjugation channels in particular those given  by $1\ten U \in G(S)$. We will substantially go beyond. 

We introduce a definition to isolate this class of extra unitaries achievable on the system via open system dynamics:

\begin{definition}[The Unitary Set]
\mbox{}
\begin{enumerate}
    \item The Lindbladian set of $\mathcal C_V(G)$, denoted by $\Lin_{\mc C_V(G)}$ is the set of Lindbladians that generate $\mc C_V(G)$: 
    $$ \Lin_{\mc C_V(G)} := \{ L: \exp(tL) \in C_V(G) \;\forall t \geq 0 \}$$

    \item The Hamiltonian set of $\mathcal C_V(G)$, denoted by $\Ham_{\mc C_V(G)}$ is the set of Hamiltonians $H$ which lie in $\Lin_{\mc C_V(G)}$ by adjoint action:
    $$ \Ham_{\mc C_V(G)} := \{ H \in \mc L_{sa}(A): H=H^*, \; i \ad_H \in \Lin_{\mc C_V(G)}\}$$

    \item The set of $\mathcal C_V(G)$, denoted by $\mc U_{\mc C_V(G)}$ is the set of unitaries whose conjugations are in the channel class:
    $$\mc U_{\mc C_V(G)} = \{U \in U(A): \Ad_U \in \mc C_V(G)  \}$$
\end{enumerate}
\end{definition}

To describe the properties of these three sets, we will need the following key lemma:

\begin{lemma}\label{difflemma}
    Let $\Phi_t \in \mc C_V(G)$ be a continuous path of channels with $\Phi_0 = \id$ which admits a  right derivative as
    \[ 
    L  := \lim_{t \to 0^+} \frac{\Phi_t - \id}{t} \pl. 
    \] 
Then we have 
    \[
    L \in \Lin_{\mc C_V(G)} \qquad \exp(sL) \in \mc C_V(G) \; \forall s \geq 0.
    \]
    Similarly, if $L$ exists and $L=0$, and the second right derivative $L_2$ exists, then $e^{tL_2}\in C_V(G)$.
\end{lemma}
\begin{proof}
    We have $$\Phi_t = \id + t L + R_t$$ with $R_t = o(t)$. We also have $$\exp(tL) = \lim_{n \to \infty} \lpara \id + \frac{t}{n} L\rpara^n.$$ Letting $\phi_t = \id + tL$ we write:
\begin{align*}
    \|\phi_t^n-\Phi_{tn}\|
    &\le \sum_{j=0}^{n-1} \|\phi_t^j(\phi_t-\Phi_t)\Phi_{t(n-j)}\| 
    \le \|\phi_t-\Phi_t\| \sum_{j=0}^{n-1} \|\phi_t\|^j \\
    &\le \|R_t\| \frac{(1+t\|L\|)^{n}-1}{t\|L\|}
    \le \|L\|^{-1} \frac{\|R_t\|}{t} e^{nt\|L\|} \pl .
\end{align*}
  For the first assertion we can assume $L \neq 0$. Thus   \[ \lim_n \|\phi_{s/n}^n-\Phi_s\| \lel 0 \pl .\] For the second assertion,we take  $R_t=\Phi_t-(\id+\frac{t^2}{2}L_2)$ which satisfies $R_t = o(t^2)$.  
We let $\phi_t := \id+\frac{t^2}{2}L_2$ which satisfies
 \[ \lim_n \phi_{\sqrt{\frac{2}{n}}t}^n \lel e^{t^2L_2} .\] 
We see
 \[ \|\phi_{\sqrt{\frac{2}{n}}t}^n-\Phi_{\sqrt{\frac{2}{n}}t}^n\| 
  \kl \frac{\|R_{\sqrt{\frac{2}{n}}t}\|}{\frac{2t^2}{n}} \pl \frac{e^{t^2\|L_2\|}}{2}  \]
converges to $0$. Now by the fact that $\mc C_V(G)$ is topologically closed and closed under composition, we see that the lemma holds.
\end{proof}

\begin{corollary}\label{lieham} If $\pm H_1, \pm H_2 \in \Ham_{\mc C_V(G)}$, then $\pm i[H_1,H_2]\in \Ham_{\mc C_V(G)}$.
In particular, if $\Ham_{\mc C_V(G)}$ is a real vector space it is a Lie algebra. 
\end{corollary}
\begin{proof}
    Let $H_1, H_2 \in \Ham_{\mc C_V(G)}$. We have $e^{\ad_{iH_1t}}e^{\ad_{iH_2t}} e^{\ad_{-iH_1t}} e^{\ad_{-iH_2t}} \in \mc C_V(G)$. Now we note that $e^{\ad_{iH_1t} }e^{\ad_{iH_2t}} e^{\ad_{-iH_1t}} e^{\ad_{-iH_2t}} = \Ad_{e^{iH_1t} e^{iH_2t} e^{-i H_1t} e^{-i H_2 t}}$. We have that $$\Ad_{e^{iH_1t} e^{iH_2t} e^{-i H_1t} e^{-i H_2 t}} = \Ad_{e^{[H_1,H_2]t^2 + O(t^3)}}.$$ And now by \cref{difflemma} we have that $i[H_1,H_2] \in \Ham_{\mc C_V(G)}$. 
\end{proof}
\begin{remark}
     If $G$ is a Lie group then $\mc U_{\mc C_V(G)}$ contains the Lie group generated by $\Ham_{\mc C_V(G)}$. In general $\mc U_{\mc C_V(G)}$ is a compact group.
\end{remark}

Note now that we have two different sets of unitaries, one set which is given to us as allowable operations $G(S) \subseteq U(EA)$ on the global space and another set $\mc U_{\mc C_V(G)} \subseteq U(A)$ which we only get after iteratively applying the open system technique of state preparation. To make this tractable it serves us to define a slicing map $(S,V) \hookrightarrow \hat S_V \subseteq \Ham_{\mc C_V(G)}.$ To define such an embedding we require two lemmas: 

\begin{proposition}\label{embed}
    Let $P_j \in \mc L_{sa}(E)$ and $Q_j \in \mc L_{sa}(A)$ for $1 \leq j \leq k$. Assume $\sum_j^k P_j \otimes Q_j \in S$ and $v \in V$ with $G = G(S)$. Then we have: \[ \om_v(T) \lel T_v \pl :=\pl   \sum_{i=1}^k \braket{v | P_j | v} Q_j \in \Ham_{\mc C_V(G)}.
    \]
\end{proposition}
\begin{proof}
    Let $T = \sum_j^k P_j \otimes Q_j$ with $v \in V$. Take $\Phi_t = \Phi^v_{\exp(iTt)}=\tr_E e^{iTt}(\ket v\bra v\otimes (\cdot)) e^{-iTt}$. Writing $L = \lim_{t \to 0^+} \frac{1}{t} \Phi_{t}$ we calculate
    \begin{align*}
    L(\rho) &= i \tr_E[T, \ket v \bra v \otimes \rho ] \lel  \sum_{j=1}^k i \braket{v|P_j|v} [Q_j, \rho] 
    \end{align*}
    By \cref{difflemma} we deduce the assertion.
\end{proof}

\begin{remark}
Using \cref{embed} we can now define a map $(S,V) \to \hat S_V$ by taking any $\sum P_j \otimes Q_j \mapsto \braket{v | P_j | v} Q_j$. Note that this embedding does not get us all Hamiltonians in $\Ham_{\mc C_V(G)}.$ In general $\hat S_V$ generates a large subspace and large Lie sub-algebra of $\Ham_{\mc C_V(G)}.$
\end{remark}
We can picture the underlying relations of the above operations by the following diagram:

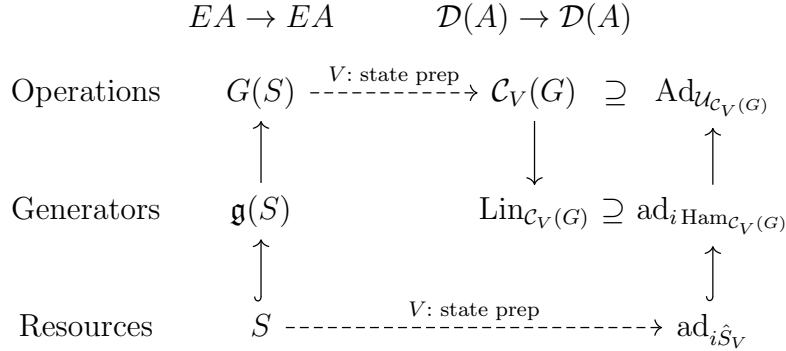
\begin{figure}[h]
\label{fig:channelchart}

\begin{tikzcd}
 & [-1 cm] EA \to EA   & \mathcal D(A) \to \mathcal D(A)  & [-1.3 cm]  \\
 [-0.5 cm]
\text{Operations} & G(S) \arrow[r, dashed, "{V:  \text{ state prep}}"] & \mathcal C_V(G) \arrow[r, "\supseteq", phantom] \arrow[d]    & \text{Ad}_{\mathcal U_{\mathcal C_V(G)}} \\
\text{Generators} & \mathfrak g(S) \arrow[u]                 & \text{Lin}_{\mathcal C_V(G)} \arrow[r, "\supseteq", phantom] & \ad_{i \Ham_{\mathcal C_V(G)}} \arrow[u]   \\ \text{Resources}
& S \arrow[u, hook] \arrow[rr, dashed, "{V: \text{ state prep}}"] & &    \ad_{i \hat S_V} \arrow[u, hook]   
\end{tikzcd}
\caption{How we go from $S$ to $\mc U_{\mc C_V(G)}$}
\end{figure}

\begin{figure}[t]
\label{fig:channeltable}

\begin{tblr}{|c|c|c|}
\hline
    Set & Definition & Purpose \\ \hline \hline
    $S$ & $S \subseteq \mc L_{sa}(EA)$ & Resource set\\ \hline
    $V$ & $V \subseteq E$ &  Preparable states \\ \hline
    $\mf g(S)$ & Smallest lie algebra containing $iS$ & Approximable generators \\ \hline
    $G(S)$ & Lie group of $\mf g(S)$ & Approximable Gates  \\ \hline
    $\mc C_V(G)$ &  $ \overline{ \{\prod_{i=1}^n \Phi^{v_i}_{U_i}: v_i\in V, U_i \in G, i \in \mathbb N\}}^{\diamond}$ &  Approximable Channels \\ \hline
    $\Lin_{\mc C_V(G)}$ &  $\{ L: \exp(tL) \in C_V(G) \; \forall t \geq 0\}$ & Generators of Channels \\ 
    \hline
    $\Ham_{\mc C_V(G)}$ & $\{ H \in \mc L_{sa}(A): i \ad_H \in \Lin_{\mc C_V(G)}\}$ &  Generators of Unitary Channels \\ 
    \hline $\mc U_{\mc C_V(G)}$ & $\{U \in U(A): \Ad_U \in \mc C_V(G) \}$ & Open System Unitaries \\ 
    \hline 
\end{tblr}
\caption{Table listing the various sets and their definitions}
\end{figure}

In the ideal case, the operations implementable by $G(S)$ will be as small as possible while $\Ham_{C_V(G)}$ will be whole Lie algebra $\mathfrak{su}(A)$, hence 
rich enough for state preparation.  We will see that with the right amount of state preparation it is in fact possible to have generators $S$ so that $\mf g(S) \sim poly(n)$ while $\mc U_{\mc C_V(G)} = SU(A)$. Often times we start with a general set of Pauli strings $S$ that is too small to achieve state-transitivity, and add a Hamiltonian $H$ so that $\mc C_V(G(S \cup \{H\}))$ is state-transitive while $G(S\cup\lset H\rset)$ still remains small. By analogy with magic T-gates \cite{bravyikitaev}, we call such an operator a \textit{magic Hamiltonian.}

\begin{definition}
    Let $S \subseteq \mc L_{sa}(EA)$ be a set of generators and $V \subseteq E$ be a set of preparable states. Let $H$ be a Hamiltonian and define $\tilde S = S \cup \{H\}$. Take $\tilde G = G(\tilde S).$ We call $H$ \textit{state magic} if $\mc C_V(\tilde G)$ is state transitive.
\end{definition}

Our general construction will start with a resource set $S$ containing no interaction between the system and environment. We will then add a magic interaction Hamiltonian $H$ which gets us state magic. The Lie group $G(S)$ is especially easy to calculate in the case where there are no interaction terms:
\begin{proposition}
    Let $S_A \in \mc L_{sa}(A)$ and $S_E \in \mc L_{sa}(E)$ be sets of self-adjoint generators. If $S = 1_E \otimes S_A \cup S_E \otimes 1_A$ with $G = G(S)$ then $\mc C_V(G)$ contains only channels given by unitary conjugation for any $V$. That is to say, $\mc C_V(G) = \Ad_{\mc U_{\mc C_V(G)}}=\Ad_{G(S_A)}$.
\end{proposition}
\begin{proof}
    We first note that $1_E \otimes S_A$ commutes with $S_E \otimes 1_A$ so that $$\mf g(S) \subseteq i\mc L_{sa}(E) \otimes 1_A \cup i1_E \otimes \mc L_{sa}(A).$$
    Now we see that $\mf g(S)$ splits into two commuting Lie algebras $\mf g_A(S) \oplus \mf g_E(S)$. We then can write any $U \in G(S)$ as $U = (1 \otimes U_A)(U_E \otimes 1)$. We can then calculate 
    $$U \lpara \ket v \bra v \otimes \rho \rpara U^* = U_E (\ket v \bra v) U_E^* \otimes U_A \rho U_A^*.$$
    Hence  $$\tr_E( U (\ket v \bra v \otimes \rho )U^* ) = U_A \rho U_A^*.$$ which is indeed a unitary channel.
\end{proof}

\section{Conditions for state magic}

We continue working with channels and Lindbladians. We will first show sufficient criteria on $\mc \Lin_{\mc C_V(G)}$ in order to achieve state transitivity, and then we will translate those conditions to criteria on an interaction Hamiltonian.

\subsection{Channel Desiderata}

We now explore what criteria are sufficient for a channel class $\mc C_V(G)$ to be state magic. Compiling such sufficiency conditions will provide us with reasonable desiderata and intuitions needed to create small resource sets with magic Hamiltonians. We first recall the following facts about Lindbladians:
\begin{fact}\label{lind_identity}
\begin{enumerate}
    \item \label{uni_change}
    For any unitary $u\in\mc U( A),$ and some jump operator $a$,
    $$L_{uau^*}=\Ad_u\circ L_a\circ \Ad_{u^*}\qquad e^{tL_{uau^*}}=\Ad_u\circ e^{tL_a}\circ \Ad_{u^*}$$

    \item \label{average_change}
    For any matrix $a$ and $b$, 
    $$\frac{1}{2}(L_{a+b}+L_{a-b})=L_a+L_b$$

    \item \label{scalar_change} For any jump operator $a$ and $\gamma\in\mb C$, 
    $$L_{\lambda a}=|\lambda|^2L_a$$

    \item \label{identity_change} $L_I=0$
\end{enumerate}
\end{fact}

Many of our examples will be built from the same scaffolding: A single non-selfadjoint jump operator, together with full Lie algebra $\mf{su} \in \Ham_{\mc C_V(G)}$, suffices for state transitivity. The following result is taken from \cite{cgj}. This result is predicted in \cite{bwsh16}. A more general variants studied recently in \cite{mved23}

\begin{theorem}[\cite{cgj} Thm 2.8, Ex 5.2]\label{ad_trans}
    Let $SU(2^n)\subseteq \mc U_{\mc C_V(G)}$. Take $a=\ket0\bra1\otimes I^{\otimes n-1}$. If $L_a \in \Lin_{\mc C_V(G)}$ then $\mc C_V(G)$ is state-transitive.
\end{theorem}
An algorithmic proof is given by Example 5.2 of \cite{cgj}. The main idea is to use coherent control to reduce the problem to the evolution of diagonal states, thereby turning state transitivity into a problem of manipulating their eigenvalues. This is closely related to the viewpoint developed in \cite{mved23}.

We can extend the above result to arbitrary non-normal jump operators: 

\begin{theorem}\label{general_trans}
    Suppose that  $SU(2^n) \subseteq \mc U_{\mc C_V(G)}$ . Let $ a_k = I^{\otimes k}\otimes a\otimes I^{\otimes (n-k-1)}$ with $a \in\mb M_2$ is non-normal, i.e. $[a, a^*]\neq 0$. If $\mc L_{a_k} \in \Lin_{\mc C_V(G)}$ then $\mc C_V(G)$ is state transitive.
\end{theorem}
\begin{proof}[Proof of \cref{general_trans}]
    We may first move the dissipative resource to the first qubit. Thus, it suffices to consider jump operator $a \otimes I^{\otimes n-1}$. Indeed, a suitable permutation unitary $u\in  U(A)$ exchanges the $k$-th qubit with the first $\Ad_u(I^{\otimes k}\otimes a\otimes I^{\otimes (n-k-1)})=a\otimes I^{\otimes n-1}$, and \cref{lind_identity} (i) shows that the corresponding conjugated dissipation is accessible.

    We next reduce the jump operator to creating the amplitude damping channels for $A=\cz^2$. By Schur decomposition, there is a unitary $v\in U(2)$ so that $\Ad_v(a)=\begin{pmatrix}
        \alpha & \gamma\\
        0 & \beta
    \end{pmatrix}=d+b$
    where $d=\begin{pmatrix}
        \alpha & 0\\
        0 & \beta
    \end{pmatrix}$ and $b=\begin{pmatrix}
        0 & \gamma\\
        0 & 0
    \end{pmatrix}=\gamma\ket0\bra1$. Since $a$ is non-normal, $\gamma\neq 0$. Using \cref{lind_identity} (i), we may now assume our jump operator is $d+b$.
    
    Let $Z=\diag(1, -1)$ be the Pauli-$Z$ matrix. Then $\Ad_Z(d+b)=d-b$ so coherence control also give access to $d-b$. By \cref{lind_identity} (ii), 
    $$\frac{1}{2}(L_{b+d}+L_{b-d})=L_b+L_d.$$
    Hence, by the Lie-Trotter product formula, $$e^{t(L_{b}+L_{d})}=\lim_{m\rightarrow\infty}\lpara e^{\frac{t}{2m}L_{d+b}}e^{\frac{t}{2m}L_{d-b}}\rpara^m$$ belongs to the closure of the accessible channel class.

    Following the proof in \cite{cgj} it suffices to  observe that the evolution preserves the diagonal and how transform  diagonal states. Indeed, if a two dimensional $\rho$ is diagonal in the computation basis, then $L_d(\rho)=0$. Moreover,  $L_b$ preserves the diagonal sub-algebra. It follows that for every diagonal $\rho$
    $$e^{t(L_d+L_b)}(\rho)=e^{tL_b}(\rho).$$
    Since $b=\gamma\ket0\bra1$,  $L_b=|\gamma|^2L_{\ket0\bra1}$, which is precisely the amplitude-damping evolution of \cref{ad_trans}, up to a rescaling of time. Since full unitary control allows us to first diagonalize any state, the the proof of \cref{ad_trans} only requires the amplitude-damping evolution on the diagonal sub-algebra. Hence the accessible evolution $e^{t(L_b+L_d)}$ can be used in place of $e^{tL_b}$ in that argument.
\end{proof}

\begin{corollary}[Channel Class Conditions]\label{ccc}
    The following two conditions are sufficient in order for $\mc C_V(G)$ to be state transitive:
    \begin{enumerate}
        \item $\mf{su}(A) \subseteq i\Ham_{\mc C_V(G)}.$
        \item $L_{a_1} \in \mc \Lin_{\mc C_V(G)}$ where $a_1 = a \otimes I^{\otimes (n-1)}$, with $a \in \mb M_2(\mb C)$ an arbitrary non-normal operator.
    \end{enumerate}
\end{corollary}

\subsection{From Channels Conditions to State Magic} We now have sufficient conditions on $\mc C_V(G)$ for state transitivity. However, our construction starts from a resource set $S$. In order to make use of the above conditions, we translate those sufficiency conditions to conditions on an interaction Hamiltonian $H \in \mc L_{sa}(EA).$ In the case of such a Hamiltonian, our group of operations can be expanded to $\tilde G = G(S \cup \{H\})$. 

\begin{remark}\label{strategy}
The first condition of \cref{ccc} on the Lie algebra can now be given for the expanded group $\mf{su}(A) \subseteq i\Ham_{\mc C_V(\tilde G)}.$ For the second condition of \cref{ccc} we want to take the continuous family of channels given by the interaction Hamiltonian: \[ \Phi^H_t(\rho) := \tr_E\lpara e^{iHt} \lpara\ketbra{v} \otimes \rho \rpara e^{-iHt} \rpara.\] We can use \cref{difflemma} to construct Lindbladians from such a path. If any of those Lindbladians satisfies the second condition of \cref{ccc} then we can conclude that the interaction Hamiltonian $H$ performs state magic. 
\end{remark}

To execute the strategy of \cref{strategy} we need to make use of the following channel approximation:
\begin{fact}
Suppose that $E=\mb C^2$. 
    Let us consider a Hamiltonian in the block matrix form 
    $H=\begin{pmatrix}
        \alpha & \gamma^*\\
        \gamma & \beta
    \end{pmatrix}$ in the computation basis. Then $$\Phi_t(\rho)=\tr_E(e^{it H}(\ket 0\bra 0\otimes\rho) e^{-itH})$$
    is of the form 
    \begin{equation}\label{channel_1}
        \Phi_t(\rho)=\Ad_{e^{it\alpha}}(\rho)+\frac{t^2}{2}L_{\gamma}(\rho)+O(t^3).
    \end{equation}
    Note that we can do the same approximation for \[ \Phi_t(\rho)=\tr_E \lpara  e^{it H}\lpara \ketbra{v} \otimes\rho \rpara e^{-itH} \rpara \] by writing $H$ in the $\left \{ \ket v, {\ket v}^{\perp}\right \} $ basis instead of the computation basis. 
\end{fact}

It is often useful to translate the Hamiltonian between block-matrix form and Pauli-vector form:

\begin{definition}
    Let $E = \mb C^2$ and $H \in \mc L_{sa}(EA)$ the \textit{Pauli vector} associated to $H$ is \[
    \vec{H}: = \lpara H_I, H_X,H_Y, H_Z \rpara.
    \]
    Given by 
    \[
    H = I \otimes H_I + X \otimes H_X + Y \otimes H_Y + Z \otimes H_Z.
    \]
\end{definition}

We know that it is paramount for us to have a Lindbladian $L_a \in \Lin_{\mc C_V(G)}$ with non-normal $a$. Using the above results, we can find an interaction Hamiltonian $H \in \mc L(EA)$ which can create the generator $L_a$ through open system dynamics. 

\begin{proposition}\label{jumplemma}
    Take $E \cong \mb C^2.$ Let $a\in \mc L(A)$ be a fixed jump operator and $H \in \mc L_{sa}(EA)$ be an interaction Hamiltonian. Take $\ket 0 \in V.$ Let $iH \in \mf g(S)$ and $G = G(S).$
    The following statements are equivalent:
    \begin{enumerate}
        \item $H$ is a Hamiltonian lift of the jump operator $a$ with respect to $\ket 0$, in the sense that $\Phi_t^H=\id+\frac{t^2}{2}L_a+O(t^3)$.
        \item $H$ is of the block-matrix form $$H=\begin{pmatrix}
            \lambda I & e^{-i\theta} a^*\\
            e^{i\theta} a & K
         \end{pmatrix}$$
        for some $\lambda, \theta\in\mb R$ and some arbitrary self-adjoint block $K=K^*$.
        
        \item \[\vec{H} = \lpara \frac{\lambda I+K}{2}, \frac{e^{i\theta}a+e^{-i\theta}a^*}{2} , \frac{e^{i\theta}a-e^{-i\theta}a^*}{2i}, \frac{\lambda I-K}{2} \rpara \]
    \end{enumerate}
    If $H$ satisfies one of the above equivalent conditions, then $L_a \in \Lin_{\mc C_V(G)}.$
\end{proposition}
\begin{proof}
    By \cref{channel_1}, we know that $\Phi_t(\rho)=\Ad_{e^{it\alpha}}+\frac{t^2}{2}L_\gamma+O(t^3)$. This is equal to $\id+\frac{t^2}{2}L_\gamma+O(t^3)$ if and only if $e^{it\alpha}=\omega_t I$, where $\omega_t\in\mb C$, $|\omega_t|=1$. Hence $\alpha=\lambda I$ for some $\lambda\in\mb R$. By \cref{lind_identity} (iii), $L_{e^{i\theta}a}=L_a$ and hence $H_{10}=e^{i\theta}a$. There is no restriction for $\beta$ so it can be an arbitrary self-adjoint block $K$. Thus (ii) follows. It is clear that (ii) is equivalent to (iii) by expanding $H$ in the Pauli basis.
\end{proof}
\begin{remark}
    If the environment is prepared in a different unit vector $\ket v\in V$ instead of $\ket 0$, the same statement hold after writing $H$ in the orthonormal basis $\{ \ket v, {\ket v}^{\perp}\}$ of $E$. 
\end{remark}
From \cref{jumplemma} we can already describe a quintessential example of a state magic Hamiltonian:
\begin{example}[Amplitude damping]
Let $E \cong \mb C^2$ and $A \cong \mb C^2$. Take $V = \{ \ket{0}_E \}$. Take
    \[H=\frac{1}{2}(Y_E\otimes X_A-X_E\otimes Y_A) \quad \mbox{or} \quad  H \lel \frac{1}{2}(X_E\ten X_A+Y_E\ten Y_A)   \pl  .\] 
Let $\tilde G = G(S \cup \{H\}).$ If $\mf{su}(A) \subseteq \Ham_{\mc C_V(\tilde G)}$ then $H$ performs state magic.
\end{example}
\begin{proof}
    From $\cref{jumplemma}$ we have that $L_a \in \Lin_{\mc C_V(G)}$ with $a = \ket{0}\bra{1}$ By \cref{ad_trans} we have that $H$ is state magic.
\end{proof}

The form of the Hamiltonian in \cref{jumplemma} can be restrictive. It is useful for us to state more general conditions under which an interaction Hamiltonian can perform state magic:

\begin{theorem}[The $\gamma$ magic theorem]\label{densitymagic}
We take $\ket 0 \in V.$ Assume $E \cong \mb C^2$ Take an interaction Hamiltonian \[H=\begin{pmatrix}
        \alpha & \gamma\\
        \gamma^* & \beta
    \end{pmatrix}.\] Take $\tilde G = G(S \cup \{H\}).$ If $\mf{su}(A) \subseteq i\Ham_{\mc C_V(\tilde G)}$ and $\gamma$ is non-normal then $H$ performs state magic.
\end{theorem}
\begin{proof}
    By \cref{channel_1}, $$\Phi_t(\rho)=\Ad_{e^{it\alpha}}+\frac{t^2}{2}L_{\gamma}+O(t^3).$$
    Now consider the channel
    $$\Psi_t:= \Ad_{e^{-it\alpha}}\circ\Phi_t
        = \Ad_{e^{-it\alpha}}\circ\lpara\Ad_{e^{it\alpha}}+\frac{t^2}{2}L_{\gamma}+O(t^3)\rpara=\id+\frac{t^2}{2}L_{\gamma}+O(t^3)$$
    because $\Ad_{e^{it\alpha}}=\id+O(t)$.
    Then after the square-root rescaling, 
    $\Psi_{\sqrt{\frac{2t}{n}}}=\id+\frac{t}{n}L_{\gamma}+O((\frac{t}{n})^{\frac{3}{2}})$ and thus 
    $$(\Psi_{\sqrt{\frac{2t}{n}}})^n\rightarrow e^{tL_{\gamma}}.$$
    Now \cref{general_trans} implies the assertion.
\end{proof}
\begin{remark}
    In the setting \cref{densitymagic} instead of taking $\ket 0 \in V$ we can take $\ket v \in V$. In this case we use the $\{ \ket v, {\ket v}^{\perp}\} $ basis and write \[H=\begin{pmatrix}
        \alpha & \gamma\\
        \gamma^* & \beta
    \end{pmatrix}\] in that basis. The results of $\cref{densitymagic}$ still hold.  
\end{remark}
We have similar result for a Hamiltonian written in the Pauli basis:
\begin{corollary}\label{ham_lift}
    Assume $E \cong \mb C^2$ and $\ket 0 \in V$. Take an interaction Hamiltonian of the form \[\vec{H} = (H_I,H_X,H_Y,H_Z).\] Take $\tilde G = G(S \cup \{H\}).$ If $\mf{su}(A) \subseteq i\Ham_{\mc C_V(\tilde G)}$ and $[H_X, H_Y]\neq 0$, then $H$ performs state magic. 
\end{corollary}
\begin{proof}
    If $[H_X,H_Y] \neq 0$ we have $\gamma$ is non-normal for the hypothesis of \cref{densitymagic}.
\end{proof}

It is interesting to note that the results of \cref{densitymagic} do not depend on the size of $G$. In fact, if we take a resource set $S$ with interaction Hamiltonian $H$, it is possible for both $\mf g(S)$ and $\mf g(S \cup \{H\})$ to be $\text{poly}(n)$. We will see such an example in \cref{spinex}. 

Moreover, it is in fact possible for $S$ to be empty, and achieve state magic with only a single interaction Hamiltonian and no resources. In this case we will need two state preparations, $\ket{0}_E, \ket{1}_E \in V$ to satisfy the transitivity desiderata in \cref{ccc}. We will see such an example in \cref{2prep}.

\section{Spin example with few  qubits and small Lie algebra}\label{spinex}

In this section we aim for an example with a polynomial Lie algebra for the `control' group $G\subset U(EA)$. Let us recall the notations 
 \[ \mathcal{Z}(1)\lel \{Z_1,...,Z_n\}\pl ,\pl \mathcal{X}(2) \lel \{X_1X_2,...,X_{n-1}X_n\} \pl,\] 
 \[ \mathcal{Y}(2) \lel \{Y_1Y_2,...,Y_{n-1}Y_n\}  
 \pl,\pl \mathcal{Z}(2) \lel \{Z_1Z_2,...,Z_{n-1}Z_n\}  
 \pl .\] 
We will identify a Lie algebra with interaction Hamiltonians with three additional qubits as follows 
  \[ S_{tot} \lel S_1 \cup S_2 \cup\{Z_{-1}Z_1Z_2\} \pl. \] 
 \begin{align*} 
  S_1  &\lel  \{X_{-2}X_{-1}X_0X_1,Y_1Y_2,X_{-2}X_{-1}X_2X_3, Y_3Y_4,...,Y_{n-1}Y_n\}  \\ 
   S_2  &\lel  \{Y_{-2}Y_{-1}Y_0Y_1,X_1X_2,Y_{-2}Y_{-1}Y_2Y_3, X_3X_4,...,X_{n-1}X_n\}  \\ 
 \end{align*}
We have adapted the well-known commuting sets (see  \cite{jpv})
  \begin{align*}
   \tilde{S}_1 & \lel  \{X_0X_1,Y_1Y_2,X_2X_3, Y_3Y_4, \dots \} \\
    \tilde{S}_2 & \lel \{Y_0Y_1,X_1X_2,Y_2Y_3, X_3X_4, \dots \} 
    \end{align*}
 by modifying the first and the third entry so they still commute with $\{Z_{-1}Z_1Z_2\}$. Indeed, adding to Pauli of the same flavor in unrelated registers does not change the  commutation relation and hence does not change the Lie algebra, see also the definition of PauLie maps in  \cite{jpv}. However, the elements in the set $S_1$ itself form a chain: Every element only anti-commutes with its neighbors, and commutes with all other elements. The Lie algebra of such a chain is known. It consists exactly of alternating produces 
  \[ X_{j}X_{j+1}Y_{j+1}Y_{j+2}X_{j+2}X_{j+3}\cdots X_{k}X_{k+1}
  \lel X_jZ_{j+1}Z_{j+2}... Z_kX_{k+1} \pl .\] 
Of course, depending on parity, the start or end can change. In case we have an odd number elements in the chain $l=2k+1$, this can be modeled on $k+1$ qubits and generates the Spin group $Spin(2k+2)$. We refer to \cite{lawson} for more information on Spin groups and the double cover 
 \[ 1 \rightarrow \zz_2 \rightarrow Spin(m) \rightarrow SO(m) \rightarrow 1 \pl .\]  
Using the model qubits from \cite{jpv} and \cite{eisert}, allows to prove the relation between $l$ and $k$. For even $n$, the chains $S_1$, $S_2$  has $l=2k$ many elements and hence the Lie groups satisfies
 \begin{equation} \label{evv} 
   G(S_1) \lel  Spin(n+2) \pl ,\pl G(S_2) \lel Spin(n+2) \pl .
    \end{equation} 
For odd $n$, we find at least
 \begin{equation} \label{odd} 
   G(S_1) \subset Spin(n+3 ) \pl, \pl G(S_2) \subset Spin(n+3 )  \pl. 
   \end{equation} 
Since the dimension of the Lie algebra $\mf{spin}(m)=\so(m)$ is $\frac{m(m-1)}{2}$, we can collect these facts in the following lemma:   

\begin{lemma}\label{sgroup} The dimension of Lie algebra $lie(S_{tot})$ is at most $(n+3)(n+2)+1$. The Lie group is contained in 
 \[ G(S_{tot})\subset \mathbb{T} \times Spin(n+3)\times Spin(n+3) \pl .\] 
\end{lemma}

\begin{lemma} Let $\xi$ be a vector such that 
 \[ \lan \xi,X\xi \ran \neq 0\pl  , \pl \lan\xi,Y\xi\ran  \neq 0 \pl ,\pl \lan \xi,Z\xi\ran\neq 0\pl \] 
and $V=\{|\xi\xi\xi\ran ,|000\ran\}$. Let $G=G(S_{tot})$. Then $C_V(G)$ contains $\Ad_{U(A)}$ and acts transitively on densities.  
\end{lemma}  

\begin{proof} Let $H=A_{-2}B_{-1}C_0D$ such that $A,B,C$ are single Paulis and $D$ is a Hamiltonian.  Then 
 \[ tr_{E} [H,\si_{-2}\ten \tau_{-1}\ten \gamma_0 \ten \rho]
 \lel  tr(A\si)tr(B\tau)tr(C\gamma) [D,\rho] \pl .\] 
By our choice of $\xi=\frac{1}{\sqrt{6+2\sqrt{3}}}\begin{pmatrix}
    \sqrt{3}+1\\1+i
\end{pmatrix}$, we may `trace out' all the single Paulis and find new Hamiltonians 
 \[ \{Z_1Z_2, X_1,Y_1,X_2X_3, Y_2Y_3\}\subset Ham_{C_v(G)} \] 
which can be added to the existing Hamiltonian. Therefore  
 \[ \{X_1,Y_1,Z_1,Z_1Z_2\} \cup \mathcal{X}(2)\cup \mathcal{Y}(2) 
 \subset Ham_{C_V(G)} \pl .\] 
Using a standard trick 
 \[ [X_jX_{j+1},Z_j] \lel -2i  Y_jX_{j+1} \pl ,\pl [Y_jX_{j+1},Y_jY_{j+1}] \lel 2i Z_{j+1} \] 
we deduce by induction and $\mathcal{Z}(1)\subset \Ham_{C_V(G)}$ as well. We refer to \cite{jpv} for the fact that this implies full unitary control.

With the environment initialized in $\ket{000}$, the only non-zero off-diagonal matrix element of $X_{-2}X_{-1}X_0X_1+Y_{-2}Y_{-1}Y_0Y_1$ from the $\ket{000}$-sector is proportional to $X_1-iY_1=2\ket1\bra0$. Hence the second-order reduction yields $L_{\ket1\bra0}\in\Lin_{C_V(G)}$ up to a positive scalar. Together with full unitary control, \cref{general_trans} and \cref{ccc} gives state transitivity.
\end{proof}
             
 \begin{rem} \label{op}{\rm With only one extra qubit and two extra vectors, we can also achieve state-transitivity and full unitary control. Indeed, let us consider 
\begin{align*}
    T_1 &=  \{X_0X_1,Y_1Y_2,X_2X_3,...,X_{n-1}X_n,Y_{n-1}Y_n\} \pl ,\pl \\ 
    T_2 &=  \{Y_0Y_1,X_1X_2,Y_2Y_3,...,Y_{n-2}Y_{n-1},X_{n-1}X_n\} \pl. 
\end{align*}
We can add the element  $Y_2Y_3Y_4Y_5$ to $T_2$ and find a Lie group 
$G(T_2\cup \{Y_2Y_3Y_4Y_5\}) \subset SU(2^{n/2+2})$.  This leads to a control group of dimension roughly $\sqrt{4^n}$. Then after trace out with $|\xi\ran$ we have full unitary control. Trace out with $|0\ran$ produces state transitivity. This indicates the subtle competition between size of the group in $SU(EA)$ and expressitivity. } \end{rem}

\section{A Magic Hamiltonian with an empty resource set}\label{2prep}
In the previous examples we start with a resource set $S$ which generates a polynomial sized lie algebra. It turns out if we allow for multiple state preparations $\ket{0}_E, \ket{1}_E \in V$, we can have a magic Hamiltonian $H$ in the case where $S$ is empty! That is to say, a single Hamiltonian $H$ suffices to for state transitivity. In this case the group $G$ of operations is given by the one-parameter unitary subgroup $G_H := \{e^{iHt}:t \in \mb R\}$.  Moreover we will see in \cref{channelmagic} that this one-dimensional unitary group can be used to approximate any channel $\Phi: \mc D(A) \to \mc D(A).$

\begin{theorem}[The $\alpha \beta \gamma$ magic]\label{abcmagic}
Take $E = \mb C^2$, and an empty resource set $S = \varnothing.$ We take an interaction Hamiltonian $H \in \mc L_{sa}(EA)$ written in the computation basis as
\[ H = \begin{pmatrix} \alpha & \gamma^* \\ \gamma & \beta \end{pmatrix}.\]
Take $V = \{\ket{0}_E, \ket{1}_E \}$ If $\mf{su}(A)\subseteq \mf g(\{\alpha,\beta\})$ and $\gamma$ is non-normal, then $H$ performs state magic.
\end{theorem}
\begin{proof}
    Both $\Phi_{e^{itH}}^{\ket{0}} = \Ad_{e^{it\alpha}} + O(t^2)$ and $\Phi_{e^{itH}}^{\ket{1}} = \Ad_{e^{it\beta}} + O(t^2)$ are in the channel class $C_V(G_H).$ By \cref{difflemma} we have that $\alpha, \beta \in \Ham_{C_V(G_H)}.$ By the assumption that $\mf g(\alpha,\beta) = \mf{su}(A)$ and \cref{lieham} we have that $\mf{su}(A) \subseteq \Ham_{C_V(G_H)}$. And now through \cref{densitymagic} we see that $H$ performs state magic.
\end{proof}
\begin{remark}
    \cref{abcmagic} can be modified to take $V = \{  \ket{v}, \ket{v}^\perp \}$ by writing $H$ in the basis according to the set $V$. 
\end{remark}

\cref{abcmagic} shows us that with just a single Hamiltonian we can act transitively on the entire state space, if we find two Hamiltonians which together are bracket generating (see also forthcoming \cite{cgj2}). 

In this case rather than starting with any resources in the resource set $S$, our resources can instead 'live inside' the magic Hamiltonian $H$. The two different state preparations $\rho \mapsto \ketbra{0} \otimes \rho$ and $\rho \mapsto \ketbra{1} \otimes \rho$ allows us to extract the resources $\alpha, \beta$ that live on the diagonal of the Hamiltonian $H.$ We can see a concrete example of this below:

\begin{example}\label{1d+2prep}
    Let $E \cong \mb C^2$ and $A \cong \mb C^2$. Let $S = \varnothing$ and $V = \{ \ket{0}_E, \ket{1}_E \}.$ Write \[H=\begin{pmatrix}
    0 & 1 & 0 & 0 \\
    1 & 0 & 1 & 0 \\
    0 & 1 & 1 & 0 \\
    0 & 0 & 0 & -1  
    \end{pmatrix}.\] By \cref{abcmagic} we have $X,Z \in \Ham_{C_V(G_H)}$ and $L_a \in \Lin_{\mc C_V(G_H)}$ where $a = \ket{0}\bra{1}.$ Now $\mf g(X,Z) = \mf{su}(A)$ and $\gamma=\ket0\bra1$ is non-normal, so we have state magic using the algorithm from \cref{ad_trans}. Remarkably, iterated applications of the two channels $\Phi_{e^{itH}}^{\ket{0}}(\rho) = \tr_E \lpara e^{iHt} \lpara \ketbra{0} \otimes \rho \rpara e^{-iHt} \rpara $ and $\Phi_{e^{itH}}^{\ket{1}}(\rho) = \tr_E \lpara e^{iHt} \lpara \ketbra{1} \otimes \rho \rpara e^{-iHt} \rpara $, together with limits, suffices to send any state $\rho$ to any other state $\sigma$.
\end{example}
We can also construct examples that work for arbitrary numbers of qubits:
\begin{example}\label{anymagic}
    Take $A \cong (\mb C^2)^{\otimes n}$ and $E \cong \mb C^2$. Take $S = \varnothing$ and $V= \{ \ket{0}_E, \ket{1}_E \}.$ Write
    $$p =\sum_{k=1}^{2^n} k \ketbra{k} - \frac{2^n+1}{2}I, \qquad
        q = \sum_{j \neq k} \ket{j}\bra{k},\; \text{ and }\; 
        a = \ket{0}\bra{1}$$
    where $p,q,a \in \mc L(A).$ Note that $p$ and $q$ are traceless and $\lie(ip, iq)=\su(2^n)$, and $a$ is non-normal. Now we write the Hamiltonian \[ H = \begin{pmatrix}p & a \\ a^* & q \end{pmatrix}.\] By \cref{abcmagic} we have that $C_V(G_H)$ is state transitive and $H$ performs state magic.
\end{example}

We can also identify an open system magic Hamiltonian for $n$ qubits.

\begin{theorem}  Let 
 \[ h\lel  \kla \begin{array}{cc} Z_1&  \frac{X_1+iY_1}{2} \\
              \frac{X_1-iY_1}{2}& X_1+\sum_j X_{j}X_{j+1}+Y_jY_{j+1}+Z_jZ_{j+1} \end{array}\mer \pl ,\]
$G=\{e^{ith}|t\in \rz\}$ and $V=\{\ket0,\ket1\}$. Then $C_V(G)$ performs density magic.              
\end{theorem}

\begin{proof} We see that off diagonal part $\gamma$ is exactly the jump operator for the dephasing channel. Thanks to Theorem \ref{abcmagic}, it suffices to observe that $\{Z_1,X_1+H\}$ are bracket generating for the $XY$ model $H=\sum_j X_{j}X_{j+1}+Y_jY_{j+1}+Z_jZ_{j+1}$, closely related to the result in \cite{Burgarth2009}, see \cref{two}. 
\end{proof} 
We see that with just a single Hamiltonian we can achieve full state transitivity by iterating state preparation to two orthogonal vectors. One Hamiltonian suffices, but only using a single Hamiltonian may be costly in terms of iteration length. We will explore the cost and complexity of state magic in \cref{lowerbounds}.

\section{Lower bounds for approximation length}\label{lowerbounds}

In this section we prove that certain states or certain densities may be considered expensive in this new model of quantum computation. The relevant definition is the approximation length  
 \[ L_{\eps}(\rho|G,V) \lel \inf\{m\}\]  
for a given density $\rho \in \mathcal{D}(A)$. The infimum is taken over all $U_1,...,U_m\in G$ and $v(1),...,v(m)\in V$  such that 
 \[ \|\Phi_{U_1}^{v(1)}\cdots \Phi_{U_m}^{v(m)}(|0\ran\lan 0|^A)-\rho\|_1<\eps \pl .\] 
We use
 \[ l_{\eps}(G,V) \lel \sup_{\rho=|v\ran\lan v|} L_{\eps}(\rho|G,V) \] 
for the worst case approximation of pure states and 
 \[ L_{\eps}(G,V) \lel \sup_{\rho\gl 0,tr(\rho)=1} L_{\eps}(\rho|G,V) \] 
for  approximation of densities. Our key technical tool are covering numbers 
 \[ N(K,\eps) \lel \inf\lset m: \exists {x_1,...,x_m},  K\subset \bigcup_{j=1}^m B_{\eps}(x_j)\rset \pl .\] 
To be more precise, we have to specify the metric $d$ for which the Balls $B_{\eps}(x)$ are defined. For unitary groups $G\subset \mz_m$ we will use the $\infty$ norm. Let us start with a well-known fact from Banach space theory, see \cite{pisier}.  For every real Banach space $X$ of real dimension $d$ 
 \begin{equation} \label{cov}
  (\frac{1}{\eps})^d \kl   N(B_X,\eps) \kl(1+\frac{2}{\eps})^d 
 \end{equation}
Here $B_X=\{x| \|x\|_{X}\le 1\}$  is the unit ball and the distance defined by the norm  on $X$.  
For the set of densities we use the space of self-adjoint matrices equipped with the trace norm $\|x\|=tr(|x|)$ in $Sh_1$. Here $Sh_1$ is the Banach space $L(A)$ equipped with the trace norm $\|x\|_1=tr(|x|)$ with unit ball $B_{Sh_1}$.  Let $D=\mathcal{D}(A)$ be the set of positive elements of trace $1$. Then we see that 
 \[ B_{Sh_1} \subset {\rm conv}(D,-D) \subset 2 B_{Sh_1} \pl .\]    

\begin{lemma}\label{dens} Let $H$ be an $n$ dimensional Hilbert space and $D=\mathcal{D}(H)$. 
 \[ \lpara\frac{1}{8\eps}\rpara^{\frac{n^2-1}{2}}\kl  N(D,\eps) \kl \lpara1+\frac{4}{\eps}\rpara^{n^2-1}\pl .\] 
Here the distance is taken with respect to the trace norm. 
   
\end{lemma}

\begin{proof} Let $N(D,\eps)$ be a covering $D\subset \bigcup_j B_{\eps}(d_j)$. Let $t_1,...,t_N$ be an $\eps/2$-net of $[0,1]$ with $M\leq 2/\eps$.
For $a\in B_{Sh_1}$ we may write 
 \[ a \lel \la D_1-(1-\la)D_2 \pl. \] 
We choose $\|d_j-D_1\|<\eps$, $\|d_j'-D_2\|<\eps$, and $|\la-t_j|<\eps$. This implies
 \begin{align*}
  \|a-t_jd_j-(1-t_j)d_j'\|_1
  & \le    \|\la D_1-t_jd_j\|_1+\|(1-\la)D_2-(1-t_j)d_j'\|_1 \\
  &\le  2|\la-t_j| + \|D_1-d_j\|_1+\|D_2-d_j'\| \kl 4\eps \pl .
   \end{align*}
This implies for $d=n^2$ that 
 \[  (\frac{1}{4\eps})^d \kl 
 N(B_{Sh_1},4\eps)\kl N(D,\eps)^2 \frac{2}{\eps} \pl .\] 
Now, we prove the upper estimate by sending $d\in D$ to $d-\frac{1}{n}$ (see also \cite{Aubrun} and \cite{Sz}). This maps preserves the $\|\cdot \|_1$ distance. In the definition of the covering numbers one may assume that the elements are in the set, $K_1\subset K_2$ implies $N(K_1,\eps)\le N(K_2,\frac{\eps}{2})$.  Now the upper estimate for the Banach space of self-adjoint matrices with trace $0$ of real dimension $d=n^2-1$ implies the assertion. 
\end{proof}

\begin{lemma}\label{vect} Let $P\subset D\subset \mz_n$ be the set of pure states. Then 
 \[ (\frac{c}{\eps})^{2n-2} \kl N(P,\eps) \kl (\frac{C}{\eps})^{2n} \pl .\]  
Here the distance is taken with respect to the trace norm. 
\end{lemma}

\begin{proof} For two pure state $\xi$ and $\eta$ we find
 \[ \||\xi\ran\lan\xi|-|\eta\ran\lan \eta| \|_1
 \kl  \||\xi\ran \lan \xi-\eta| \|_1+ \||\xi-\eta\ran\lan \eta| \|_1 
  \kl 2 \|\xi-\eta\| \pl .\] 
Conversely, 
 \[ \||\xi\ran\lan\xi|-|\eta\ran\lan \eta| \|_1
 \geq \||\xi\ran\lan\xi|-|\eta\ran\lan \eta|\xi\ran\lan \eta|\|_1
 \lel \| \xi-\lan \xi,\eta\ran \eta\| \pl .\] 
Hence the results for the complex Hilbert spaces apply. This is an upper bound for the distance in the projective space $\mb{CP}^{n-1}\simeq S^{2n-1}/S^1$. The lower bound follows from the quotient map $S^{2n-1}\rightarrow\mb{CP}^{n-1}$ by the $S^1$-action. Since $\mb{CP}^{n-1}$ has real dimension $2n-2$, we obtain the lower bound.
\end{proof}

We will now the general state a result on the lower bound of the approximation length, a variation of the \cite{jpv}.

\begin{prop} \label{lww} Let $|A|=2^n$ and 
$G\subset U(EA)$ be a group such that 
 \[ N(G,\eps) \kl (1+\frac{a}{\eps})^d \pl .\] 
and $|V|=k\geq 2$. For $\eps\le \eps_0$ the following estimates hold  
 \[ L_{\eps}(G,V) \gl \min(c(\eps)\frac{2^{2n}}{dn  \ln (1+a)},\frac{2^{2n}}{6\ln |V|}) \pl \] 
and 
 \[ l_{\eps}(G,V) \gl \min(c(\eps)\frac{2^n}{dn\ln (1+a)},\frac{2^n}{6\ln |V|}) \pl \]  
\end{prop}

\begin{proof} Fix $\eps>0$. Let $m=L_{\eps/2}(G, V)$ be the maximal approximation length. For a given density $\rho$ let $(v(1),U_1),...,(v(m),U_m)$ be choices in  the approximation for $\frac{\eps}{2}$. Let $S$ be an $\frac{\eps}{4m}$ net for $G$. Then we note that 
 \[ \|\Phi_{U_j}^{v_j}-\Phi_{U}^{v_j}\|_\diamond
 \kl \|ad_{U_j}-ad_{U}\|_\diamond \kl 2 \|U-U_j\|_{op} \pl .\]
This shows that we find $V_1,...,V_m\in S$ such that  
 \[ \|\Phi_{U_1}^{v(1)}\cdots \Phi_{U_m}^{v(m)}-\Phi_{V_1}^{v(1)}\cdots \Phi_{V_m}^{v(m)}\|_\diamond
  \kl \sum_{j=1}^m \|\Phi_{U_j}^{v(j)}-\Phi_{V_j}^{v(j)}\|_\diamond \kl \frac{\eps}{2} \pl .\]
In other words the set 
 \[ \Delta \lel \{\Phi_{V_1}^{v(1)}\cdots \Phi_{V_m}^{v(m)}(|0\ran\lan 0|^A)| V_1,...,V_m\in S, v(1),...,v(m)\in V\} \] 
is an $\eps$-net for either all densities or all rank one densities. In the first case we deduce that 
 \begin{align*}
  \frac{2^{2n}-1}{2}\ln\frac{1}{8\eps}
  & \kl \ln |\Delta| 
  \kl m (\ln|S|+\ln |V|) \\
  &\kl m \big (d\ln (1+\frac{4Cam}{\eps})+\ln |V|\big) \\
  &\kl 2d m\ln m+ dm \ln (1+\frac{4Ca}{\eps})+ m\ln |V|  \\
  &\kl 3 \max\{2dm \ln m,dm \ln(1+\frac{4Ca}{\eps}),m\ln |V|\} \pl .
  \end{align*}
We assume that $8\eps \le \frac{1}{e}$. Then we deduce that either
\[ 2^{2n}\kl \frac{8}{\ln \frac{1}{8\eps}} 2dm \ln m \pl ,  \mbox{ or } 2^{2n}\kl  
  \frac{\ln(1+\frac{4Ca}{\eps})}{\ln \frac{1}{8 \eps}} 8dm\pl,  \mbox{ or } 2^{2n} 
  \kl 8m \ln |V| \pl .\]
We deduce the assertion from Lemma \cite[Lemma 5.12]{jpv} in this cases. For pure states we use the lower bound $\frac{2^n}{4}\ln \frac{c}{\eps}$ instead. Replacing $\eps$ by $2\eps$ givess the states estimate.
\end{proof} 

It remains to find upper bounds for the groups in consideration. 

\begin{lemma} Let $h$ be the open system magic Hamitonian and $G=\{e^{ith}|t\in \rz\}$. Then 
 \[ N(G,\eps) \kl (1+\frac{4}{\eps})^{2(2^n+1)} \pl. \]
\end{lemma}

\begin{proof}
After diagonalizing $h$, the group $G_h=\lset e^{ith}:t\in \mb R\rset$ is contained in the unit ball of the diagonal algebra $\ell_{\infty}^{2^n+1}$. Hence
 \[ N(G,\eps) \kl N(B_{\infty},\frac{\eps}{2}) \kl (1+\frac{4}{\eps})^{2(2^n+1)} \] 
using the formula for the real dimension of this complex Hilbert space. 
\end{proof}
\begin{rem} If we could control the size of the spectrum of $h$ this estimate could possibly be improved considerably.
\end{rem} 

\begin{cor}
    There are densities with approximation length $\gl c(\eps)\frac{2^n}{n}$ for the open system magic Hamiltonian. 
\end{cor}

We will now address state preparation for the Spin example  of section 3. The group in question is a product group $G\subseteq G_1\times G_2\times G_3$. Clearly 
 \[ N(G_1\times G_2\times G_3,\eps) \kl N(G_1,\frac{\eps}{3})  N(G_2,\frac{\eps}{3}) N(G_3,\frac{\eps}{3}) \]
follows from 
 \[ \| u_1v_1w_1-u_2v_2w_2\|  \kl  \|(u_1-u_2)v_1w_1\| + \|u_1(v_1-v_2)w_1\| + \|u_1v_1(w_1-w_2)\| \]
and the fact that unitaries are contractions. For the Spin group representation as Clifford matrices we recall from Corollary 5.8 of \cite{jpv} that
 \[ N(Spin(m),\eps) \kl (\frac{Cm}{\eps})^{2m(2m-1)}  \pl . \]
This means we can apply Proposition \ref{lww} for
 \[  a=C(n+3)\le 2Cn \quad \mbox{and} \quad  d=2(2m(2m-1))+1\le 8\times 16 n^2 \pl .\]  
\begin{theorem} Let $G=G(S_{tot})$ be  the group from \cref{spinex}. Then $C_{V}(G)$ acts transitively on the space of densities and contains $ad_{U(A)}$. There are densities with 
  \[ L_{\eps}(\rho|G,V) \gl c(\eps) \frac{2^{2n}}{n^3\ln(1+n)}  \]
and pure states such that 
 \[  L_{\eps}(\ket\eta\bra\eta|G,V) \gl c(\eps) \frac{2^{n}}{n^3\ln(1+n)}\pl . 
 \]
\end{theorem}

\section{Channel Universality}\label{channelmagic}
In the previous sections we were considered with \textit{state magic}, which gets us a class of channels transitive on the space of densities $\mc D(A)$. In this section we will be considering \textit{channel magic.} In this case we will have a single Hamiltonian $H$ which can generate a channel class which contains all channels $\mc D(A) \to \mc D(A)$. 

\subsection{From State Magic to Channel Magic}

To generate all channels we will need to use Stinespring dilation:
\begin{fact}[Stinespring dilation \cite{stinespring}]
    Let $A \cong (\mb C^2)^{\otimes n}$ and $F \cong (\mb C^2)^{\otimes 2n}.$ For any channel \( \Phi:\mc D(A) \to \mc D(A)\) there is a unitary $U \in SU(FA)$ so that \[\Phi(\rho) = \tr_F \lpara U\lpara \ket{0}_F \bra{0}_F \otimes \rho \rpara U^* \rpara.\]
\end{fact}

Throughout the rest of this section we will use $F$ to represent the Stinespring ancilla. We will continue to use $E$ to represent the usual preparation environment, and will combine the two.

Our general strategy will be to generate all unitaries in $SU(FA)$ so that we can use Stinespring dilation to get all channels $\Phi: \mc D(A) \to \mc D(A).$ To generate these unitaries we will use an extra environment $E$. We recall the definition of channel class, but now using an extra environment $E$ on top of the space $FA$.

\begin{definition}[Channel Class]
\mbox{}
\begin{enumerate}
    \item Take a system $A \cong (\mb C^2)^{\otimes n}$, Stinespring ancilla $F \cong (\mb C^2)^{\otimes 2n}$, and environment $E \cong (\mb C^2)^{\otimes m}.$ Let $V \subseteq E$ be a set of unit vectors in $E$. Furthermore let $G \subseteq U(EFA)$ be a group of untaries on the total space $EFA$. For any $\ket v \in V$ and $U \in G_{EFA}$ we can write the channel \[\Psi^v_U :\rho \mapsto \tr_E \lpara U\lpara \ket{v}\bra{v} \otimes \rho \rpara U^*\rpara: \mc D(FA) \to \mc D(FA).\]
    We define the channel class $\mc C^{FA}_V(G)$ as the smallest set of channels which contains the set $\{ \Psi^v_U: v \in V, \, U \in G\}$, is topologically closed under the diamond norm, and closed under composition.

    \item The Lindbladian set is defined as: \[\Lin_{\mc C_V^{FA}(G)} := \left \{ L:e^{tL} \in \mc C_V^{FA}(G) \right\}.\] Note that here we have $L: \mc L(FA) \to \mc L(FA).$

    \item The Hamiltonian set is defined as: 
    \[\Ham_{\mc C_V^{FA}}(G):= \left\{ H \in \mc L_{sa}(FA): H = H^*, \, i \ad_H \in \Lin_{\mc C_V^{FA}(G)}\right\}.\]

    \item The unitary set is defined as: 
    \[\mc U_{\mc C_V^{FA}(G)} := \left\{U \in U(FA): \Ad_{U} \in \mc C^{FA}_V(G)\right \}\]
\end{enumerate}
\end{definition}

If the set of Hamiltonians is the Lie algebra $\mf{su}(FA)$ then we can generate all channels $\mc D(A) \to \mc D(A)$:

\begin{lemma}\label{chlemma}
    Take a system $A \cong (\mb C^2)^{\otimes n}$, Stinespring ancilla $F \cong (\mb C^2)^{\otimes 2n}$, and environment $E \cong (\mb C^2)^{\otimes m}.$ Let $V \subseteq E$ be a set of vectors in $E$. If $G \subseteq U(EFA)$ is such that $\mf{su}(FA) \subseteq i\Ham_{\mc C_V^{FA}(G)}$, then we can write any channel $\Phi: \mc D(A) \to \mc D(A)$ as \[ \Phi(\rho) = \tr_F \lpara \Psi(\ket{0}_F\bra{0}_F \otimes \rho)\rpara, \quad \Psi \in \mc C_V^{FA}(G). \] 
\end{lemma}
\begin{proof}
    This proposition follows directly from Stinespring dilation.
\end{proof}

Now we can show that a single Hamiltonian suffices to generate all channels!

\begin{theorem}[Channel Magic] \label{chmagic}
    Take a system $A \cong (\mb C^2)^{\otimes n}$, Stinespring ancilla $F \cong (\mb C^2)^{\otimes 2n}$, environment $E \cong (\mb C^2)$ and preparable vectors $V = \{ \ket{0}_E, \ket{1}_E \}.$ Take the Hamiltonian $H \in \mc L_{sa}(EFA)$ written as \[H = \begin{pmatrix} \alpha & \gamma \\ \gamma^* & \beta \end{pmatrix}\] where $\alpha,\beta,\gamma \in \mc L(FA).$ 
    If $\mf{g}(\{\alpha, \beta\}) = \mf{su}(FA)$ then we can write any channel $\Phi: \mc D(A) \to \mc D(A)$ as \[ \Phi(\rho) = \tr_F \lpara \Psi(\ket{0}_F\bra{0}_F \otimes \rho)\rpara, \quad \Psi \in \mc C_V^{FA}(G_H). \]  
\end{theorem}
\begin{proof}
    By the first-order argument in the proof of \cref{abcmagic}, $\alpha,\beta\in\Ham_{C_V^{FA}(G_H)}$, where $G_H=\lset e^{itH}:t\in\mb R \rset\subseteq U(EFA)$. Since $\mf g(\lset \alpha, \beta\rset)=\su(FA)$, \cref{lieham} gives $\su(FA)\subseteq i\Ham_{\mc C_V^{FA}(G_H)}$. Hence $\Ad_{U(FA)}\subseteq \mc C_V^{FA}(G_H)$. The conclusion now follows from \cref{chlemma} and Stinespring dilation.
\end{proof}
\begin{remark}\label{approxremark}
    We define two channels, $\Psi^{\ket{0}}_t, \Psi^{\ket{1}}_t: \mc D(FA) \to \mc D(FA)$, by
    \begin{align*}
        \Psi^{\ket{0}}_t(\rho) &= \tr_E \lpara e^{iHt}\lpara \ket{0}_E\bra{0}_E \otimes \rho \rpara e^{-iHt}\rpara\\
        \Psi^{\ket{1}}_t(\rho) &=  \tr_E \lpara e^{iHt} \lpara \ket{1}_E\bra{1}_E \otimes \rho \rpara e^{-iHt}\rpara
    \end{align*}
    Recall that any channel $\Psi \in \mc C_V^{FA}(G_H)$ can be approximated in the sense that $\forall \eps> 0, \, \exists k \in \mb N$, such that \[ \norm{\Psi - \Psi^{\ket{i_1}}_{t_1} \Psi^{\ket{i_2}}_{t_2}\dots \Psi^{\ket{i_k}}_{t_k}}_{\diamond}  \leq \eps. \]
\end{remark}
Combining \cref{chmagic} and \cref{approxremark}, the protocol has a particularly simple form. The Stinespring ancilla $F$ is initialized once in $\ket0_F$, while before each application of $e^{itH}$, the preparation environment is initialized in either $\ket0_E$ or $\ket1_E$. After the resulting evolution on $FA$, the ancilla $F$ is traced out. These operations suffice to approximate every channel on $A$. One Hamiltonian suffices!

\begin{example}
    Let $A=\mb C^2$, $F=\mb C^4$, $E = \mb C^2$.
    Take, as in \cref{anymagic}
    $$p = \sum_{k=1}^{8} k \ket{k}\bra{k} - \frac{9}{2}I, \qquad
        q = \sum_{k \neq j}^{k,j \leq 8} \ket{k} \bra{j}, \quad\mbox{and}\quad
        a = \ket 0 \bra 1$$
    where $p,q,a \in \mc L(FA).$
Now if we write a Hamiltonian $H \in \mc L(EFA)$ as \[H = \begin{pmatrix} p & a \\ a^* & q \end{pmatrix}\] then by \cref{chmagic} we have that any channel $\Phi: \mc D(A) \to \mc D(A)$ can be written as \[ \Phi(\rho) = \tr_F \lpara \Psi(\ket{0}_F\bra{0}_F \otimes \rho)\rpara, \quad \Psi \in \mc C_V^{FA}(G_H).\]
\end{example}
 
 \subsection{Cost Estimates for Channel Magic} Just as state magic can be a costly operation, so can the approximation of channels. If we want to measure the cost of this approximation, we will need a length function on the channel class. To define a length function as in \cref{lowerbounds} we need to create a channel class that is topologically closed and closed under compositions:
 
\begin{definition}
     Take a system $A \cong (\mb C^2)^{\otimes n}$, Stinespring ancilla $F \cong (\mb C^2)^{\otimes 2n}$ and environment $E \cong (\mb C^2)$. We take a set of preparable vectors $V \subseteq E$ and $W \subseteq F.$ We also fix a group of untaries $G \subseteq U(EFA).$ To any $\Psi \in \mc C_V^{FA}(G)$ and $w \in W \subseteq F$ we can define the channel $\Phi_{\Psi}^w: \mc D(A) \to \mc D(A)$ by \[\Phi_{\Psi}^w(\rho) = \tr_F\lpara \Psi( \ket w \bra w \otimes \rho) \rpara .\]
     
     We define the extended channel class $\widehat{\mc C}_{W,V}^A(G)$ as the smallest set of channels $\mc D(A) \to \mc D(A)$ which contains the set $\{ \Phi^w_\Psi: w \in W, \, \Psi \in \mc C_V^{FA}(G) \},$ is topologically closed and closed under composition.
 \end{definition}
 \begin{remark}
     We have that the unitaries in $U(FA)$ form a special case of our extended channel class as:
     \[\mc C^A_W\lpara\mc U_{\mc C^{FA}_V(G)}\rpara \subseteq \widehat{\mc C}^A_{V,W}.\] Moreover if $\mf{su}(FA) \subseteq i\Ham_{\mc C^{FA}_V(G)}$ then $\mc C^A_W\lpara\mc U_{\mc C^{FA}_V(G)}\rpara$ contains all channels $\Phi: \mc D(A) \to \mc D(A)$ by \cref{chlemma}.
 \end{remark}
\begin{remark}
    We see that any $\Phi \in \widehat{\mc C}_{V,W}^A(G)$ can be approximated by repeated application of one of the following operations:
    \begin{enumerate} 
 \item[i)] Prepare $|0 \ran \lan  0|^B\ten \rho$; 
 \item[ii)] Apply a unitary $U\in U(BA)$;
 \item[iii)] Apply the channel $\psi_{\xi}(\rho)  \lel |\xi\ran\lan \xi|\ten tr_E(\rho)$ for $\xi \in E$
 \item[ iv)] Apply the channel  $\psi_{\xi}(\rho)  \lel |\xi\ran\lan \xi|\ten tr_F(\rho)$ for $\xi\in F$. 
  \item[v)] Apply the channel $\Psi(\rho)= |0\ran\lan 0|\ten tr_B(\rho)$.  
 \item[vi)] Trace out $tr_{B}\ten id$. 
 \end{enumerate} 
\end{remark}

\begin{definition}
    For a channel $\Phi:\mc D(A)\rightarrow\mc D(A)$, define its $\eps$-approximation length by 
    $$L_\eps(\Phi|G, V, W):= \inf\lset m:\norm{\Phi-\Phi_{\bf U}^{\omega, \bf v}}_\diamond<\eps\rset$$
    where the infimum is over $\omega\in W$, $v_1, \dots, v_m\in V$, and $U_1, \dots, U_m\in G$. Let the worst case approximation length be 
    $$L_\eps^{ch}(G, V, W)=\sup_{\Phi\in\cha(A)}L_\eps(\Phi|G, V, W)$$
\end{definition}

\begin{lemma}\label{cov} Let $\cha(d)=\lset \Phi:\mc D(A)\rightarrow \mc D(A): A\simeq \mb C^d, \Phi \text{ is CPTP}\rset$ be the set of all channels on a Hilbert space of dimension $d$. There exist universal constants $c, C, \eps_0>0$ such that for $0<\eps<\eps_0$,
$$\lpara\frac{C}{\eps}\rpara^{cd^4}\leq N(\cha(d), \eps)\leq \lpara1+\frac{4}{\eps}\rpara^{d^4-d^2}$$
\end{lemma} 

\begin{proof} We refer to \cite{Pau, Pis} for the definition and properties of cb-norm. We consider space of completely bounded $*$ preserving maps. This space has real dimension $d(d-1)\times 2d^2$. For every unital completely positive map $u$ we can consider $u-id$ in the subspace of dimension $2(d^4-d^3-d)$ of $cb$-norm $\le 2$. By a standard Banach space-volume argument (see \cite{pisier}) we deduce 
 \[ N(\cha(d),\eps)\kl (1+\frac{4}{\eps})^{2(d^4-d^3-d)} \pl .\]  
For the lower estimate we consider an arbitrary $*$-preserving  $u$ in the unit ball. It can be written as  
 \[ u(x) \lel {\mc V}^*x{\mc W} \pl ,\pl u(x)\lel \mc{W}^*x\mc V \] 
holds for contractions $\mc V,\mc W$. Then 
 \[ 2u(x) \lel (\mc V+\mc W)^*x(\mc V+\mc W)+(\mc V-\mc W)^*x(\mc V- \mc W) -2\mc{V}^*x\mc{V}-2\mc{W}^*x\mc{W} \] 
shows that $u(x)=\psi(x)-\phi(x)$ with 
 \[ \psi(1) \kl 4 \pl ,\pl \phi(1)\kl 2 \pl .\] 
This allows us to write
 \[ u(x) \lel a^*\hat{\psi}(x)a-b^*\hat{\phi}(x)b \] 
such that $\hat{\psi}$ and $\hat{\phi}$ are unital completely positive and $\|a\|\le 2$, $\|b\|\le \sqrt{2}$. We fix $\eps $ nets for these unit balls  and $\Delta$ an $\eps$ net for all $ucp$ maps. Let $\tilde{u}=a_s^*\hat{\psi}_ta_s-b_l^*\hat{\phi}_lb_l$. Then 
  \begin{align*}
 \|u-\tilde{u}\|_\diamond
 &  \kl \|a-a_s\|\|a\|+|a\|\|a-a_s\|+4\|\hat{\psi}_t-\hat{\psi}\|_\diamond
   + 2\sqrt{2}\|b-b_l\| + 2\|\hat{\phi}-\hat{\phi}_l\|_\diamond 
   \\
   &\le   4\eps+4\eps + 2\sqrt{2}\eps + 6\eps \kl 17 \eps \pl .
     \end{align*}  
This shows that 
 \[ (\frac{2}{17\eps})^{2d^3(d-1)} \kl 
 N(B_X,17 \eps)\kl |\Delta|^2 (\frac{C}{\eps})^{2d^2} \pl .\] 
This implies a lower bound for covering number of $\cha(d)$ using the real Banach space  of $*$-preserving maps with the $cb$-norm. \end{proof}

The same approximation argument as in \cref{lowerbounds} now applies.
\begin{theorem}\label{cha_lower}
    Let $G\subseteq U(EFA)$ satisfies $$N(G, \eps)\leq \lpara 1+\frac{a}{\eps}\rpara^d$$
    and let $V\subseteq E$ be a finite set of preparable unit vectors. Assume that $\hat C_V^A(G)$ is channel universal. Then for sufficiently small $\eps>0$,
    $$L_\eps^{ch}(G, V, \lset\ket0_F\rset)\geq \min\lset c(\eps)\frac{2^{4n}}{dn\ln(1+a)}, \frac{2^{4n}}{6\ln|V|} \rset.$$
\end{theorem}

\section{Conclusion}

One interaction Hamiltonian suffices!
In our most economical construction, a single Hamiltonian, together with repeated preparation and trace-out of a one-qubit environment, is enough to obtain universal state preparation. This raises the possibility of using the same open-system architecture as a lower dimensional variational ansatz. For example, one can consider variational quantum eigensolver problems of the form
 \[ l_{\theta v}(\rho,O) \lel tr(e^{i\theta h}(|v\ran \lan  v|\ten \rho)e^{-i\theta h}(1\ten O)) \pl .\] 
Since now we are dealing with a one parameter set traditional in a one parameters search (strictly speaking two parameter search) should be useful in optimizing the underline quantum machine learning algorithm. In this case the Haar measure should be replaced by the so-called Bohr measure. This certainly requires further investigation. Even for the larger Spin groups from section 4, the variational quantum eigensolver (VQE) could in favorable cases avoid Barren plateaus thanks to the formula 
  \[ \ez \;\text{Var}_{\theta}(|v\ran\lan v|\ten \rho,1\ten O) \lel \frac{P_{\mathfrak{g}} (|v\ran\lan v|\ten \rho), P_{\mathfrak{g}}(1 \ten O)}{\rm \dim(\mathfrak{g})} \pl    \]
in \cite{diaz2023, Larocca2025} where the expression for the projection $P_{\mathfrak{g}}$ is explained. Here one has to assume that either the density or the observable belongs to Lie algebra $\mathfrak{g}$ which is calculated explicitly in section 4. The fact that Lie algebra has dimension $O(n^2)$ rather then $4^{n}-1$ together with expressivity (implementing all unitary channels) certainly requires further investigation. 

The concrete examples presented in this paper should be viewed as an invitation to identify small Lie algebras for given resource of interacting Hamiltonian. Our results show few state preparation in the environment can make a huge difference. In the sense of \cite{jpv} this means the open system operation contain the magic which is usually attributed to $T$-gate. Our particular choice of resources, as in \cite{jpv}, stems from the fact that many current platforms for quantum computation can perform pulses, the unitary groups of local Pauli strings.

Although the resulting open-system models are universal for state preparation and channels, this universality does not make the synthesis problem easy. The lower bound obtained in this paper shows that worst-case approximation lengths remain exponential even when the underlying coherent control group has only polynomial dimension. These estimates may be viewed as a first step toward a complexity theory for state and channel synthesis in open quantum systems.

We do not expect the constructions presented here to be optimal with respect to the number of auxiliary qubits, preparable environment states, or coherent control parameters. Determining the minimal resources required for open-system state and channel universality remains an interesting open problem. More broadly, the abundance of structured Pauli-string models in quantum information suggests that many further examples of open-system magic remain to be discovered.

\section{Appendix}
Here we provide a concrete example of bracket generating with two elements. 
\begin{proposition}\label{two}
Let
 \[ H \lel \sum_j X_jX_{j+1}+Y_jY_{j+1}+\delta Z_{j}Z_{j+1} \pl. \]
The Lie algebra generated by  $\{Z_1,X_1+H\}$ is $\mathfrak{su}(2^n)$.
\end{proposition} 

\begin{proof}
For the calculation, we have to recall the standard commutator relations
 \begin{align*}
 ZX&=-iY \pl ,\pl ZY \lel iX \pl ,\pl XY \lel -iZ \pl,  \\
 XZ &=iY \pl,\pl YZ \lel -iX \pl ,\pl YX \lel iZ \pl .
 \end{align*}
The aim is now to increase the list of terms which can be produce from commutators. We need some notation 
 \[ a_j \lel X_jX_{j+1}+Y_jY_{j+1} \pl ,\pl b_j \lel Y_jY_{j+1}+\delta Z_jZ_{j+1} \pl,\pl c_j \lel X_jX_{j+1}+\delta Z_jZ_{j+1} \pl .\] 
Then we find  
 \begin{align*}
  [Z_1,X_1+H] &=[Z_1,X_1+a_1] \lel -2iY_1-2iY_1X_2+2iX_1Y_2 \pl =:\pl -2i\alpha  \pl ,\\
  [Z_1,[Z_1,X_1+H]] &= 4(X_1+X_1X_2+Y_1Y_2) \lel 4(X_1+a_1) \pl ,
 \end{align*}
In the next we take the corresponding commutator of these terms
 \begin{align*}
 [\alpha,X_1+a_1] &= 2iZ_1+2iZ_1X_2 +[Y_1+Y_1X_2-X_1Y_2,X_1X_2]+[Y_1+Y_1X_2-X_1Y_2,Y_1Y_2] \\
 &= 2iZ_1+2iZ_1X_2 +2iZ_1X_2+2iZ_1 -2iZ_2 -2iZ_2+2iZ_2 \\
 &= 2i(2Z_1+2Z_1X_2-Z_2) \pl=:\pl 2i(2Z_1+\beta) \pl .
 \end{align*}
We repeat this procedure with $\beta=2Z_1X_2-Z_2$:     
 \begin{align*}
 [2Z_1X_2-Z_2,X_1+a_1] &= -2iY_1X_2+[2Z_1X_2-Z_2,X_1X_2]+[2Z_1X_2-Z_2,Y_1Y_2] \\
 &= -2iY_1X_2 -4iY_1+2iX_1Y_2 + -2iY_1X_2 \\
 &= -2i(2Y_1X_2+2Y_1-X_1Y_2) \pl .
 \end{align*}
Note that $\al=Y_1+Y_1X_2-X_1Y_2$ and hence $X_1Y_2$ in the Lie algebra $\mf l$ generated by $Z_1$ and $X_1+H$. Another commutator with $Z_1$ leads to the following list
 \[ \{Z_1,X_1Y_2,Y_1Y_2,Y_1X_2+Y_1,2Z_1X_2-Z_2,X_1+X_1X_2,\} \pl .\] 
Then $[X_1Y_2,[X_1Y_2,[Y_1X_2+Y_1]]=4Y_1$. This means we have full control in the first qubit and $H\in \mf l$. It is well-known that this is enough for universal control \cite{Burgarth2009}. Let us spell this out. We find all three linear combinations
 \[ \{X_1X_2+\delta Z_1Z_2, X_1X_2+Y_1Y_2, Y_1Y_2+\delta Z_1Z_2\} \subset \mf l\]
Taking linear combinations we find $X_1X_2$, $Y_1Y_2$ and $Z_1Z_2$. Then we consider $[a_1,[a_1,H]]$, $[b_1[b_1,H]]$ and $[c_1[c_1,H]]$ and find the same term with a shifted index. Thus by induction we have
 \[ \{X_1,Y_1,Z_1\} \cup \mathcal{X}(2)\cup\mathcal{Y}(2) \cup \mathcal{Z}(2) \subset \mf l \pl .\]
 Then \cite[Theorem 3.6 (ii)]{jpv}  yields the assertion.  \end{proof}

\bibliographystyle{alpha}
\bibliography{reference}
\bigbreak
\end{document}